\documentclass[letterpaper, 10pt, conference]{ieeeconf}      % Use this line for a4 paper

\IEEEoverridecommandlockouts                              % This command is only needed
\usepackage{times}
\usepackage{latexsym,amsfonts,amsbsy,amssymb}
\usepackage{amsmath}
\usepackage{graphicx}
\usepackage{cite}
\usepackage{stmaryrd}
\usepackage{multirow}
\usepackage{tabularx,booktabs}
\usepackage{xcolor}
\usepackage{color}

\usepackage{soul}
\usepackage{balance}\balance

\usepackage{float}

\usepackage{algorithm}
\usepackage{algpseudocode}

\usepackage{tikz}
\usetikzlibrary{arrows,automata}
\usetikzlibrary{calc}
\usetikzlibrary{shapes.geometric,shapes}
\usetikzlibrary{positioning}

\tikzstyle{startstop} = [rectangle, rounded corners, minimum width=3cm, minimum height=0.8cm,text centered, draw=black]
\tikzstyle{process} = [rectangle, rounded corners,minimum width=3cm, minimum height=.8cm,text centered, draw=black]
\tikzstyle{decision} = [diamond, minimum width=2cm, minimum height=1cm, text centered, draw=black]
\tikzstyle{arrow} = [->,>=stealth]
\tikzstyle{blank} = [node distance=1cm]
\tikzstyle{block} = [rectangle, draw, fill=blue!20, text centered, rounded corners, minimum height=1em]
\tikzstyle{line} = [draw, -latex']
\tikzstyle{cloud} = [draw, ellipse,fill=red!20, node distance=4cm,minimum height=1em]

\title{\bf Reactive Supervisory Control of Open Discrete Event Systems}
\author{Alireza~Partovi,~\IEEEmembership{Student~Member,~IEEE,}
        and Hai~Lin,~\IEEEmembership{Senior~Member,~IEEE}
	\thanks{This work was supported in part by the National Science Foundation under Grant ECCS-1253488, Grant IIS-1724070, and Grant CNS-1830335, and in part by the Army Research Laboratory under Grant W911NF-17-1-0072.}
	\thanks{The authors are with the Department of Electrical Engineering, University of Notre Dame, Notre Dame,
		IN, 46556 USA. {\tt\small apartovi@nd.edu, hlin1@nd.edu}  \vspace{1mm}}
		}% <-this % stops a space
		
\newtheorem{theorem}{Theorem}
\newtheorem{lemma}{Lemma}
\newtheorem{proposition}{Proposition}
\newtheorem{definition}{Definition}

\newtheorem{remark}{Remark}
\newtheorem{example}{Example}

\newtheorem{problem}{Problem}

\begin{document}
\maketitle
\thispagestyle{empty}
\pagestyle{empty}

\begin{abstract}
The conventional Wonham-Ramadge supervisory control framework of discrete event systems  enforces a closed discrete event system to generate correct behaviors under certain environments, which can be captured by an appropriate plant model. Nevertheless, such control methods cannot be directly applied for many practical engineering systems nowadays since they are open systems and their operation heavily depends on nontrivial interactions between the systems and the external environments. These open systems should be controlled in such a way that accomplishment of the control objective  can be guaranteed for any possible environment, which may be dynamic, uncertain and sometimes unpredictable. In this paper, we aim at extending the conventional supervisory control theory to open discrete event systems in a reactive manner. Starting from a novel input-output automaton model of an open system, we consider control objectives that characterize the desired input-output behaviors of the system, based on which a game-theoretic approach is carried out to compute a reactive supervisor that steers the system to fulfill the specifications regardless of the  environment behaviors. We present a necessary and sufficient conditions for the existence of such a reactive supervisor. Furthermore, illustrative examples are given throughout this paper to demonstrate the key definitions and the effectiveness of the proposed reactive supervisor synthesis framework.

%The proposed reactive supervisory control scheme is illustrated on data-access system with  non-interference security specification. The unsupervised system is shown to be insecure, in a sense that an intruder, as the  environment, can deduce confidential information by monitoring the system input-output behavior. A reactive supervisory control is then designed to enforce the system to behave securely by respecting a non-interference  specification.
\end{abstract}
% \begin{IEEEkeywords}
% Discrete event systems, open systems, supervisory control theory, reactive controller synthesis, game theory.
% \end{IEEEkeywords}

\section{Introduction}
%\noindent{\it Background:} 
%\IEEEPARstart{D}{iscrete} 
Discrete event systems (DESs) refer to as a class of dynamical systems that possess a discrete state space and evolve in response to abrupt occurrences of certain qualitative changes, called {\it events}.  Due to the fact that operation of many engineering systems nowadays, ranging from intelligent manufacturing systems to transportation networks, is governed by the sequential executions of certain control actions and hence shows great event-driven features, DESs have become  useful in practice in recent years \cite{wainer2016discrete}.

A fundamental research question in DESs is to  design  a feedback controller, called a {\it supervisor}, to drive a DES to achieve certain desired formal properties. The theory of supervisory control of DESs was first introduced by  Ramadge and Wonham  \cite{ramadge1989control}, where the DES was modeled by a finite automaton, and the specifications were expressed in regular languages. Later on, the supervisory control theory for DES was further extended to variety of formal specifications, such as $\omega$-regular languages \cite{kumar1992supervisory} and temporal logic formulas \cite{jiang2006supervisory}.
%, and  different control-theoretic architectures, such as partial-observation control \cite{cieslak1988supervisory}, decentralized  control \cite{komenda2015coordination}.

% \cite{ramadge1989some,kumar1992supervisory} and temporal logic formulas \cite{lin1994reachability,jiang2006supervisory}, and  different control-theoretic architectures, such as partial-observation control \cite{cieslak1988supervisory}, decentralized  control \cite{yoo2002general,komenda2015coordination}.
%{\color{red} Others .... timed specifications, distributed, decentralized, diagnosis ... }
We argue that the majority of existing results on DES supervisor synthesis are  only suitable for  {\it closed systems}, since the interactions between the plant and environments are assumed to be known and can be fully captured by the  uncontrolled DES plant. 
Nevertheless, this assumption may not be appropriate  for many modern engineering systems, for instance, web-service security systems \cite{hashemian2005graph},  robotic manipulation systems \cite{partovi2018reactive_robotic}, since they are {\it open systems}, that are directly exposed to  dynamic and uncertain environments.  
Furthermore, the control specification for open systems, namely {\it reactive specifications}, often are required to be guaranteed with respect to all possible environment behaviors, which  goes beyond the traditional supervisory control theory of closed DESs \cite{ramadge1989control}. %\cite{wonham2015supervisory}.

%Since performance of practical engineering systems shall be influenced by not only the internal states, but also the external working environment; therefore, we need frameworks beyond closed systems to characterize such interactions between systems and their working environments.

To bridge the gap between conventional DES supervisory control theory and open systems, recent years have witnessed efforts devoted to the supervisory control of input/output DES. 
Non-blocking output supervisors were synthesized for DES with outputs in \cite{mahdavinezhad2008supervisory}. Authors  in \cite{ushio2016nonblocking} and \cite{yin2017supervisor} studied the supervisor synthesis technique for Mealy automata with non-deterministic output functions.
Authors in \cite{nke2015control} proposed a new interpretation of the I/O transitions, based on which a controller that enforced determinism and non-blockingness of the closed-loop system was designed.
The control objectives in the aforementioned papers  were  to restrict the controllable input events such that  the system output behavior meets a specification. The supervised plant is therefore, a closed-system and can not accept reactive specifications expressing the  interactions between  the plant and environment.

On the other hand, in the computer sciences literature, {\it reactive synthesis} approaches have been pursued to construct an open system, namely {\it reactive module}, to  satisfy a given specification, regardless how the environment behaves, \cite{pnueli1989synthesis,church1957applications, finkbeiner2016synthesis}. 
The majority of these works, however, either did not consider an internal dynamic model for the reactive module, or assumed that its  behavior can be encoded in the specification \cite{jiang2006supervisory}, which can dramatically  increase  the computational complexity of the synthesis procedure.
Exceptions  are  \cite{kupferman2000open}, where it was  assumed that all the plant output events are controllable and hence, no controllability constraints imposed by the plant were taken into consideration.

A connection between reactive synthesis and supervisory control of DES  was conducted in \cite{EHLERS2014222,ehlers2017supervisory}, which mainly addressed the conditions under which one design framework can be converted to the other.  
In this work, we have a different goal that is to study supervisory control for a class of open DESs with reactive specifications. 
%We aim to develop  a correct control policy that renders the system behaviours to satisfy the specifications regardless of  behaviors of the environment.

Towards this aim, we  first propose an {\it open DES} model whose behaviors  essentially  depends on the internal model of the system and its interactions with an external dynamic environment. Upon this model, we consider regular language specifications defined over the system's input-output behaviors, and then develop  a game-theoretic design approach for the supervisor such that the controlled system achieves the  the specifications regardless of  behaviors of the environment. Our basic idea is to construct a  two-player game among  the   environment, and a supervisor  representing all the controllable system outputs. It turns out that synthesizing the appropriate reactive supervisor can be reduced to finding a winning strategy for the supervisor player.

The rest of this paper is organized as follows. In
Section~\ref{sec:Open_DEC}, we introduce the open DES model and define its recognized languages. A necessary and sufficient condition for the existence of a reactive supervisor is presented in Section~\ref{sec:reactive_sup_design}. An illustrative example is utilized throughout the paper to explain the key definitions and the effectiveness of the proposed game-theoretic supervisor design method.
%In this paper, we present the key ideas and an  overview of the correctness proof, and we leave the details to our technical report \cite{partovi2018reactive} due to space limitations.
%The proposed reactive supervisory control scheme is illustrated on a data access system  with non-interference  specification in Section \ref{sec:non-interference_example}. 

\section{Open Discrete-event Systems} \label{sec:Open_DEC}

To formally define open DESs, we first review the following  notations. For a given finite set (alphabet) of {\it events} $\Sigma$, a finite {\it word} $w=\sigma_0\sigma_1\ldots \sigma_n$, $n \ge 1$, is a finite sequence of elements in $\Sigma$,  for all $\sigma_i \in \Sigma$, and $0 \le i \le n$. We denote  the length of $w$ by $\vert w \vert$. { Let $w$, and $u$ be finite words, $w \cdot u$ is their \textit{concatenations}.}
%If it is clear, we $wu$ instead 
The notation $2^\Sigma$ refers to power set of $\Sigma$, that is, the set of all subsets of $\Sigma$.
A set deference is  $\Sigma-A=\{x \mid x \in \Sigma, x \not \in A\}$.
We  denote  $\Sigma^*$  the set of all finite words including the empty word $\epsilon$. A subset of $\Sigma^*$ is called a {\it language} over $\Sigma$. The {\it prefix-closure} of a language $\mathcal{L} \subseteq \Sigma^*$, denoted as $\overline{\mathcal{L}}$, is the set of all {\it prefixes} of words in $\mathcal{L}$, i.e., $\overline{\mathcal{L}}=\{s\in\Sigma^*|(\exists t\in\Sigma^*)[st\in \mathcal{L}]\}$. $\mathcal{L}$ is said to be {\it prefix-closed} if $\overline{\mathcal{L}}=\mathcal{L}$. We denote the set of all non-negative integer as $\mathbb{N}$.
%The subset $\mathcal{L} \subseteq \Sigma^*$   defines $*-$ over alphabet set $\Sigma$. 

\begin{definition}{ 
The uncontrolled open DES plant is defined by $\mathcal{P}=(Q_p,\Sigma_p,\Sigma_x,\Sigma_y, q_{p0},\delta_p,\gamma_p,Q_{pm})$, where 
$Q_p$ is the finite set of plant states,  $\Sigma_p$ is the finite set of internal events which is  partitioned into  disjoint sets of  controllable events $ \Sigma_c$, and uncontrollable events $\Sigma_{uc}$. The finite sets $\Sigma_x=\{x_1,x_2,\dots,x_n\}$ and $\Sigma_y=\{y_1,y_2,\dots,y_m\}$, respectively representing   the external environment or the input, and the plant output event sets. $q_{p0}\in Q_p$ denotes the initial state.  $\delta_p: Q_p\times \Sigma_x \times \Sigma_p \to Q_p$ is the  transition function of the plant.  $\gamma_p: Q_p \times \Sigma_p  \to \Sigma_y$ is the output function  which assigns deterministically an output symbol to pair of the state and  events, and  $Q_{pm} \subseteq Q_p$ is the set of mark states.
% \begin{itemize}
% \item $Q_p$ is the set of plant states. 
% \item $\Sigma_p$ is the finite set of internal events which is  partitioned into  disjoint sets of  controllable events $ \Sigma_c$, and uncontrollable events $\Sigma_{uc}$.
% \item $\Sigma_x$ is the external environment or input set.
% \item $\Sigma_y$ is the plant output set.
% \item  $q_{p0}\in Q_p$ is the initial state.
% \item $\delta_p: Q_p\times \Sigma_p \times \Sigma_x \to Q_p$ is the partial transition function of the plant.

% \item  $\gamma_p: Q_p \times \Sigma_p  \to \Sigma_y$ is the output function  which assigns deterministically an output symbol to pair of the state and  events.
% \item $Q_{pm} \subseteq Q_p$ is the set of mark states.
% \end{itemize}

% {\color{red}
% To be added:
% \begin{itemize}
%     \item Open Discrete system is a input-enabled DES, meaning $\delta_p(q,\epsilon,x)!$ for all $x \in \Sigma_x$, and $q \in Q_p$.
%     \item check carefully open DES def and converting to  normal DES.
% \end{itemize}
% }

The transition function $\delta_p$, and $\gamma_p$ are respectively extended to a partial function on  $\delta_p: Q_p \times (\Sigma_x \times \Sigma_p)^*\to Q_p$, and $\gamma_p: Q_p\times \Sigma^*_p \to \Sigma_y$  in a standard way. The notation  $\delta_p(q,x,\sigma)!$ means the function $\delta_p$ is defined, and is non-empty at state $q \in Q_p$, for $x \in \Sigma_x$ and $\sigma \in \Sigma_p$. 

Without lost of generality, we assume the events sets are  disjoints,  and all the state $q \in Q$ are reachable from initial state $q_{p0}$.}
%We assume all  set of inputs is  possible in each states, i.e., $\gamma(q)=\Sigma_x$ for all $q\in Q_p$. 
% Ali: we can not handle the multi output in supervisor design but need it for LTL specification, solve it!!!!
%We further assume $\Sigma_x$ events are not capable of triggering a state-transition in the plant, although if there exist such event, it  can  be added to $\Sigma_{uc}$. 
%and $\delta_p(q_p,\omega_{xp})!$ means $\delta_p(q_p,\omega_{xp})$ is defined for $q_p \in Q_p$, and $\omega_{xp} \in (\Sigma_p \times \Sigma_x)^*$ %\cite[p.~17]{Hopcroft:2007:IAT:1454320}.
Furthermore, an open DES is required to  process every external input event. Therefore, all states $q_p \in Q_p$ have to be  input-enabled, that is $\delta_p(q_p,x,\sigma)!$, for some $\sigma \in \Sigma_p$, and  all $x \in \Sigma_x$. %Note that, this requirement is not restrictive, since for any state, and input event, we can always assume $\delta_p(q_p,x,\epsilon)
%!$, and $\gamma_p(q_p,\epsilon)=\epsilon$.
%We call $\mathcal{P}$ a closed system if $X=\emptyset$, and otherwise it is called an open DES. %The size of an open DES is $\mathcal{P}$ is defined by $|\mathcal{P}|=|Q_p|+|\delta_p|$.
\end{definition}

\begin{remark}
The proposed  open DES formalism is syntactically  similar to  an I/O automaton  of \cite{lynch1988introduction} and the interface automaton introduced in \cite{de2001interface}. In a similar fashion,  every state in an open DES is allowed to be receptive toward all possible  external input events, however, here  the output behavior is defined to be a mapping function from the internal system behavior which itself is influenced by the  environment. This condition is important to be considered in reactive synthesis of open systems since the environment may restrict the plant output behavior. Furthermore, the trace of input and output action on the I/O  and  interface automata are not necessarily prefix-closed \cite{lynch1988introduction}, which is a required property  in reactive synthesis \cite{schmuck2018relation}. We will later show that input-output event trace of an open DES is prefix-closed.
\end{remark}

\begin{remark}
In our setup,  we  consider deterministic open DES that is required to have a deterministic transition and output functions. The legal behavior of this system can be adequately expressed in terms of a finite regular  language. We believe similar to conventional non-deterministic DES \cite{hopcroft2008introduction}, non-deterministic open DESs  can  be transformed into a deterministic one that accepts the same language. 
%The legal behavior of this system can be adequately expressed in terms of a language specification.
%it is known that conventional  DES with non-deterministic transition function  can always be transformed into a deterministic one with the same language \cite{hopcroft2008introduction}. 
\end{remark}
%We define a history a finite sequence $h=(q_0,x_0,\sigma_0),  \dots ,(q_n,x_n,\sigma_{n})$, where $q_i\in Q_p$, $\sigma_{i} \in \Sigma_p$,  $x_{i} \in \Sigma_x$, and it is accepted by  $\mathcal{P}$ if $q_0 \in q_{p0}$ and  $\delta_p(q_i,x_i,\sigma_i)!$  for all $ i \in \{0,\dots,n \}$.
%Let's denote $\mathcal{H}(\mathcal{P})$  a set of all histories accepted by $\mathcal{P}$. 

Let's denote the set of all the events in the plant as $\Sigma_e=\Sigma_x \times \Sigma_p \times \Sigma_y$.
 %We define the  projection functions as $\mathrm{P}_{\alpha}=\Sigma^*_e \to \Sigma^*_\alpha$, where the index $\alpha$ belongs to a set $\{p,x,y\}$, respectively representing projection to internal event set $\Sigma_p$,  input event set $\Sigma_x$, and  output event set $\Sigma_y$.
 Projection function to input-output event, internal, and  input event sets  are respectively denote as 
$\mathrm{P}_{xy}=\Sigma^*_e \to (\Sigma_x \times \Sigma_y)^*$, $\mathrm{P}_{xp}=\Sigma^*_e \to (\Sigma_x \times \Sigma_p)^*$.
They  inductively are defined as followings and can  be extended to the sets. %$\mathrm{P}_p((\epsilon,\epsilon,\epsilon))=\epsilon$,
%      $\mathrm{P}_{xp}((\epsilon,\epsilon,\epsilon))= \mathrm{P}_{xy}((\epsilon,\epsilon,\epsilon))=(\epsilon,\epsilon)$, 
%      $\forall w \in \Sigma^*_e$, and $(x,\sigma,y)\in \Sigma_e$:
%     $\mathrm{P}_x(w(x,\sigma,y))=\mathrm{P}_x(w).x$,
%     $\mathrm{P}_p(w(x,\sigma,y))=\mathrm{P}_p(w).\sigma$,
%     $\mathrm{P}_y(w(x,\sigma,y))=\mathrm{P}_y(w).y$,
%     $\mathrm{P}_{xp}(w(x,\sigma,y))=\mathrm{P}_{xp}(w).(x,\sigma)$,
%     $\mathrm{P}_{xy}(w(x,\sigma,y))=\mathrm{P}_{xy}(w).(x,y)$.
\begin{itemize}
    \item $\mathrm{P}_p((\epsilon,\epsilon,\epsilon))=\epsilon$,
     $\mathrm{P}_{xp}((\epsilon,\epsilon,\epsilon))= \mathrm{P}_{xy}((\epsilon,\epsilon,\epsilon))=(\epsilon,\epsilon)$, 
     $\forall w_e \in \Sigma^*_e$, and $(x,\sigma,y)\in \Sigma_e$:\\
    $\mathrm{P}_x(w_e \cdot (x,\sigma,y))=\mathrm{P}_x(w_e) \cdot x$,\\
    $\mathrm{P}_p(w_e \cdot (x,\sigma,y))=\mathrm{P}_p(w_e) \cdot \sigma$,\\
    $\mathrm{P}_y(w_e \cdot (x,\sigma,y))=\mathrm{P}_y(w_e) \cdot y$,\\
    $\mathrm{P}_{xp}(w_e \cdot (x,\sigma,y))=\mathrm{P}_{xp}(w_e) \cdot (x,\sigma)$,\\
    $\mathrm{P}_{xy}(w_e \cdot (x,\sigma,y))=\mathrm{P}_{xy}(w_e) \cdot (x,y)$.
\end{itemize}

The extended input-output language of $\mathcal{P}$ captures the generated internal  and the corresponding output behavior of the plant for any possible  environment behavior. It can be formally  defined as follows.

\begin{definition}
The plant's  extended input-output language, denoted by $\mathcal{L}_e(\mathcal{P}) \subseteq \Sigma^*_e$, is defined recursively by:    
\begin{itemize}
\item $\left(\epsilon,\epsilon,\epsilon \right ) \in \mathcal{L}_e(\mathcal{P})$, and
\item for any $ w_e \in \mathcal{L}_e(\mathcal{P})$, and $(x,\sigma,y) \in  ( \Sigma_x \times \Sigma_p \times \Sigma_y )$\\ then 
$w_e \cdot(x,\sigma,y) \in \mathcal{L}_e(\mathcal{P})$, \\iff  $\delta_p(\delta_p(q_{p0},\mathrm{P}_{xp}(w_e)),x,\sigma)!$,\\ and $y = \gamma_p(\delta_p(q_{p0},\mathrm{P}_{xp}(w_e)),\sigma)$.
\end{itemize}
\end{definition}

Similarly, the marked extended input-output language of the plant can be defined as: 
$\mathcal{L}_{e,m}(\mathcal{P})= \{ w_e \in \mathcal{L}_e(\mathcal{P}) \mid \delta_p(q_{p0},\mathrm{P}_{xp}(w_e)) \in Q_{pm}\}.$ 
In the design of reactive supervisor, the input-output language of the plant is  our interest.  Input-output language of the plant $\mathcal{P}$ is $\mathcal{L}_{io}(\mathcal{P})=\{w \in ( \Sigma_x \times \Sigma_y )^* \mid w \in \mathrm{P}_{xy}(\mathcal{L}_e(\mathcal{P}))\}$, and the  input-output marked language is $\mathcal{L}_{io,m}(\mathcal{P})=\{w \in ( \Sigma_x \times \Sigma_y )^* \mid w \in \mathrm{P}_{xy}(\mathcal{L}_{e,m}(\mathcal{P}))\}$.
% \begin{definition}
% Let's denote the input-output language of the plant $\mathcal{P}$ by $\mathcal{L}_{io}(\mathcal{P})$, and the  input-output marked language  by  $\mathcal{L}_{io,m}(\mathcal{P})$. They are respectively  defined by: 
% \begin{itemize}
% \item $\mathcal{L}_{io}(\mathcal{P})=\{w \in ( \Sigma_x \times \Sigma_y )^* \mid w \in \mathrm{P}_{xy}(\mathcal{L}_e(\mathcal{P}))\}$,
% \item $\mathcal{L}_{io,m}(\mathcal{P})=\{w \in ( \Sigma_x \times \Sigma_y )^* \mid w \in \mathrm{P}_{xy}(\mathcal{L}_{e,m}(\mathcal{P}))\}$.
% \end{itemize}
% \end{definition}
The  internal language of $\mathcal{P}$, and its  marked internal language are defined  over execution of internal events, $\Sigma_p$,  and respectively are given as $\mathcal{L}(\mathcal{P})= \{w_p \in \Sigma^*_p \mid w_p \in \mathrm{P}_{p}(\mathcal{L}_e(\mathcal{P}))\}$, and $\mathcal{L}_m(\mathcal{P})= \{w_p \in \Sigma^*_p \mid w_p \in \mathrm{P}_{p}(\mathcal{L}_{e,m}(\mathcal{P}))\}$.
% Internal generated language of $\mathcal{P}$ is defined  over execution of internal events, $\Sigma_p$.
% \begin{definition}
% The  internal language of $\mathcal{P}$, and its  marked internal language are respectively is denoted by $\mathcal{L}(\mathcal{P})$, and $\mathcal{L}_e(\mathcal{P})$, and are defined by:
%  \begin{itemize}
% \item $\mathcal{L}(\mathcal{P})= \{w_p \in \Sigma^*_p \mid w_p \in \mathrm{P}_{p}(\mathcal{L}_e(\mathcal{P}))\}$,
% \item $\mathcal{L}_m(\mathcal{P})= \{w_p \in \Sigma^*_p \mid w_p \in \mathrm{P}_{p}(\mathcal{L}_{e,m}(\mathcal{P}))\}$.
% \end{itemize}
% \end{definition}
% A  state $q \in Q_P$ is called \textit{deadlock} state if $(q,\sigma,q') \not \in \delta_p$ for all $q \neq q'$ and $q' \in Q_p$, and history $(q_0,\gamma_I(q_0),\sigma_0,\gamma_p(q_0,\sigma_0)),  \dots ,(q_n,\gamma_I(q_n),\sigma_{n},\gamma_p(q_n,\sigma_n))$ is a \textit{deadlock history} if the last state, $q_n$, is a deadlock state.
We illustrate the introduced open DES model and the generated languages in the following example.

\begin{example} \label{exp:introduce_openDES}
Let's consider an open DES shown in Fig.~\ref{fig:example1}, where  $\Sigma_p=\{\sigma_{c},\sigma_{u} \}$,  $X=\{x_1,x_2\}$, and  $Y = \{y_1,y_2 \}$. 
An edge in the model is in the form of  $x\sigma \slash y$, where $\sigma \in \Sigma_p$ is the enabled internal event, and $x \in \Sigma_x$ represents the environment event, and $y \in \Sigma_y$ is the generated output event. Multiple labels over an edge indicates  multiple enabled transitions.
For example, a word in extended input-output language is $(\epsilon,\epsilon,\epsilon)(x_1,\sigma_{u},y_1)(x_1,\sigma_{c},y_2) \in \mathcal{L}_{e}(\mathcal{P})$,  and the projection of  it into the input-output language and internal language  of the plant respectively are $ (\epsilon,\epsilon)(x_1,y_1) (x_1,y_2) \in \mathcal{L}_{io}(\mathcal{P})$, and  $ \epsilon \sigma_u \sigma_{c} \in \mathcal{L}(\mathcal{P})$.
%\hfill\ensuremath{\square}
\end{example}  
 \tikzstyle{state}=[rectangle,thick,draw=black!75,
   			  fill=black!20,minimum size=2mm]
  \tikzstyle{state_p}=[circle,thick,draw=black!75,
 			  fill=black!20,minimum size=2mm,inner sep=1mm]  			  
 \begin{figure}[!h]  
\centering
\begin{tikzpicture}[shorten >=1pt,node distance=2.2cm,on grid,auto, bend angle=20, thick,scale=1, every node/.style={transform shape}] 
	\node[state_p,initial left,initial text=] (q_0)   {$\scriptstyle q_0$};
     \node[state_p] (q_1)   [right=of q_0,yshift=0cm] {$\scriptstyle q_1$};
  
    \node[state_p] (q_2) [accepting,right=of q_1,xshift=0cm] {$\scriptstyle q_2$};
            
%	\node[state_p,accepting] (q_3) [below=of q_2,yshift=0cm,xshift=0cm] {$q_3$};
	\path[->]
			(q_2) edge [loop above]   node     [right,align=center, pos=0.6]     {\begin{tabular}{c} $\scriptstyle x_1\sigma_c\slash y_2$\\[-1mm]$\scriptstyle  x_2\sigma_u\slash y_1$\end{tabular}} (q_2)
        (q_0) edge []   node    [below,align=center]     {\begin{tabular}{c}$\scriptstyle x_1\sigma_u\slash y_1$\\[-1mm]$\scriptstyle  x_1\sigma_c\slash y_2$ \end{tabular}} (q_1)
        (q_0) edge [bend left=40]   node [sloped,anchor = north ]       {$ \scriptstyle x_2 \sigma_c \slash y_2$} (q_2)
		(q_1) edge  [ left=35] node [below,align=center ] {\begin{tabular}{c}$\scriptstyle  x_1 \sigma_c \slash y_2$\\[-1mm]$\scriptstyle  x_2\sigma_u\slash y_1$\end{tabular}} (q_2)
%		(q_1) 	edge  [right=45] node [below,align=center] {$\scriptstyle \{x_1,\sigma_{u} \}\slash y_2$\\$\scriptstyle \{x_1,\sigma_{1} \}\slash y_1$} (q_3)
	  %  (q_2) edge [loop above]   node         {$\scriptstyle \{\Sigma_x-\{x_2\},\sigma_{1} \}\slash \{y_2\}$} (q_2)
    %    (q_2) edge []   node [left ,align=center]       {$ \scriptstyle \{x_2,\sigma_{u}\}\slash y_1$\\$ \scriptstyle \{x_2,\sigma_{2}\}\slash y_2$} (q_3) ;
    ;
			\end{tikzpicture}
	\caption{Open discrete event system for Example 1.}
    \label{fig:example1}
		\end{figure}
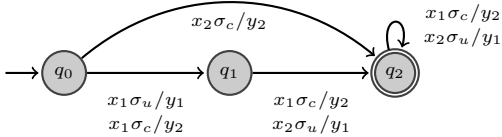

%Ali , at which, the plant's output can be observed by $\mathcal{M}(\sigma)$. 
% The input-output behavior of the plant is defined as $\mathcal{B}\left(\mathcal{P}\right) ({\Sigma_x})^\infty \to \Sigma_y^\infty$.
% Let $w= x_0 x_1 \dots x_k \in \Sigma_x^*$, $k \ge 0$, be a finite word observed from the  environment, the plant's corresponding output  $\mathcal{B}\left(\mathcal{P}\right)(w)$,  is a word in a form of $y_0 y_1 \dots y_k \in \Sigma_y^*$. 
% We use a notation of \textit{sequential input-output relationship} to formalize this behavior \cite{petreczky2009control}.    

\section{Sequential Input-output Behavior }
In this section we aim to characterize  the input-output behaviour of the proposed open DES as a reactive module such that it recognizes a regular language specification. 
The input-output behaviour  a reactive module is typically implemented over a Mealy or Moore transducer \cite{eilenberg1974automata}. 
%A reactive specification defines a relationship between an  input and an output set such that for any subset of inputs this relationship holds. In the reactive synthesis formalism, the input-output relationship is typically defined over a Mealy or Moore transducer \cite{eilenberg1974automata}.

\begin{definition}%[Mealy transducer]
A deterministic {\it Mealy transducer} is a tuple $\mathcal{T} =(Q_k,  \Sigma_x, \Sigma_y, q_{k0},\delta_k, \gamma_k,Q_{km})$,  where $Q_k$ is the finite set of states, $q_{k0} \in Q_k$ is the initial state, $\Sigma_x$ and $\Sigma_y$ are respectively  disjoint set of input and output event sets, 
$\delta_k: Q_k \times \Sigma_x \rightarrow Q_k$ is the transition function,  $\gamma_k: Q_k \times \Sigma_x \rightarrow \Sigma_{y}$ is the state output function, and $Q_{km} \subseteq Q_k$ is the set of marked states. 
%We assumes $\epsilon \in \Sigma_x$ and the output function for empty input is empty $\epsilon=\gamma_k(q,\epsilon)$.
An accepted run in $\mathcal{T} $ over a { finite input word $w_x=x_0x_1\dots x_n \in \Sigma^*_x$,} is denoted as $r=q_0q_1\dots q_n$, where $q_0=q_{k0}$, and $q_{i+1}=\delta_k(q_{i},x_{i})$ for $i\in \{0,\dots,n-1\}$. The run $r$ over input word $w_x$ generates a finite word $L(w_x) \in (\Sigma_x \times \Sigma_y)^*$ such that $L(w_x)_i= (\gamma_k(q_i,x_i), x_i)$ for $i\in \{0,\dots,n-1\}$. The input-output language or simply language of $\mathcal{T} $ is $\mathcal{L}(\mathcal{T} ) = \{  L(w_x) \mid \forall w_x \in \Sigma^*_x \}$, and the marked language  $\mathcal{L}_m(\mathcal{T} ) = \{  L(w_x) \mid \delta_k(q_{k0},w_x) \in Q_{km},  \forall w_x \in \Sigma^*_x  \}$.
A transducer is called \textit{input-complete} if $\delta_k(q,x)! $ for any $q \in Q_k$ and $\forall x \in \Sigma_x$.
\end{definition}
%\subsection{Sequential Input-output Behavior }
The proposed open DES, $\mathcal{P}$, at the state $q \in Q_p$,   reads the environment input $x \in \Sigma_x$, executes an available internal event $\sigma \in \Sigma_p$,  transits according to the transition function $\delta_p$, and generates the output events through $\gamma_p(q,\sigma)$. 
We use a notation of \textit{sequential input-output relationship} to formalize this behavior.
%, and furthermore, characterize the conditions  which  guarantees there exitence of Mealy or Moore transducer, $\mathcal{T}$,  that  accepts input-output behaviour of $\mathcal{P}$, i.e., $\mathcal{L}_{io}(\mathcal{P})=\mathcal{L}(\mathcal{T})$ \cite{Lawson2005}.

\begin{definition} %[{Sequential input-output relation}]
A  function $\mathcal{R}:{\Sigma^*_x} \to {\Sigma^*_y}$ is called sequential input-output relation  if the following conditions hold:  
\begin{description}%[label={}]
\item[$C_1:$] $\mathcal{R}(\epsilon)= \epsilon$, and for all $w_x  \in \Sigma^*_x$, $\mathcal{R}(w_x)$ is defined.
\item[$C_2:$] For all $w_x \in \Sigma^*_x, |w_x|=|\mathcal{R}(w_x)|$.
\item[$C_3:$] $\mathcal{R}$ is prefix preserving, i.e.,
for all $w_x,v_x \in \Sigma^*_x$, if $v_x \in \overline{w_x}$ then $ \mathcal{R}(v_x) \in \overline{\mathcal{R}(w_x)}$.
%\item[$C_3:$] $\mathcal{R}$ is prefix preserving, i.e.,
%for all $w_x,v_x \in \Sigma^*_x$, if $v_x \in \overline{w_x}$ then $ \mathcal{R}(v_x) \in \overline{\mathcal{R}(w_x)}$.
% explain Ali, use "Hierarchical Control of Discrete-Event Systems " paper
\end{description}
\end{definition}
%The above  conditions ensure that $\mathcal{R}$ is non-empty  and  preserves the same length for input and output words.
%A partial function $\mathcal{R} : \Sigma^*_x  \to \Sigma^*_y$ is prefix-preserving if $dom( \mathcal{R})$ is prefix-closed, and for all $u, w \in  dom( \mathcal{R})$, if $u \in   \overline{w}$ then $ \mathcal{R}(u) \in \overline{\mathcal{R}(w)}$ \cite{Lawson2005}. 
If input-output  behaviour of an open DES, $\mathcal{P}$, satisfies the properties of  sequential input-output relationship, then   
there exists a finite-state Mealy transducer, $\mathcal{T}$, such that  $\mathcal{L}_{io}(\mathcal{P})=\mathcal{L}(\mathcal{T})$ \cite{Lawson2005}.
% \begin{lemma}\cite{Lawson2005}
% A function is the output response function of a
% finite Mealy machine if and only if the input-output relation 
% \end{lemma}

\begin{proposition} \label{pro:seq_in_out}
The input-output language of an open DES  $\mathcal{P}$ respects the sequential input-output relation $\mathcal{R}$.
 \end{proposition}

\begin{proof}
Since $(\epsilon,\epsilon) \in \mathcal{L}_{io}(\mathcal{P})$, and all the states of $\mathcal{P}$ are input-enabled, $C_1$ trivially holds.
We first prove $C_3$ holds. 
Consider any $w_x,v_x  \in \Sigma^*_x$, such that $v_x \in \overline{w_x}$. Note that since all the states in plant $\mathcal{P}$  are input-enabled, the plant recognizes any $w_x \in \Sigma^*_x$, and hence, there should exist $w_e \in \mathcal{L}_{e}(\mathcal{P})$ such that $\mathrm{P}_{x}(w_e)=w_x$.  Furthermore since  $\mathcal{L}_{e}(\mathcal{P})$ is prefix-closed,  and  plant recognizes $v_x$, there should exist  some $v_e \in \mathcal{L}_{e}(\mathcal{P})$ such that $v_e \in \overline{w_e}$ and $\mathrm{P}_{x}(v_e)=v_x$. 
Since the plant transition and output functions are deterministic, if $v_e \in \overline{w_e}$, we have $\mathrm{P}_{y}(v_e) \in \mathrm{P}_{y}(\overline{w_e})$, and thus $C_3$ holds.

Moreover, since $\delta_p$, and $\gamma_p$  are deterministic functions, for any environment input word, the plant generates an output word with the same length, i.e., $\forall w_e \in \mathcal{L}_{e}(\mathcal{P})$, we have $|\mathrm{P}_{x}(w_e)|=|\mathrm{P}_{y}(w_e)|$, and therefore $C_2$ holds. 
%Let's denote $w_y=\mathrm{P}_{y}(w_e)$, and $v_y=\mathrm{P}_{y}(v_e)$.
%Since $v_e \in \overline{w_e}$, 
%we have $ \mathrm{P}_{y}(\overline{w_e})=w_y$
%  For any $ w_x \in \Sigma^*_x$, there always exist a $L_w \in \overline{\mathcal{L}_e(\mathcal{P})}$, such that $\mathrm{P}_{x}(L_w)=w$. For any $w \in \Sigma^*_x$, and $L_w$, let's define a relationship function $\mathcal{R}_p(w)=\{w_y \in  \Sigma^*_y \mid \mathrm{P}_{x}(L_e)=w \leftrightarrow w_y= \mathrm{P}_{y}(L_e) \}$. The relation $\mathcal{R}_p$  meets $C_2$, since $\gamma_p$ is a deterministic function, and therefore for any $w_x \in \Sigma_x$, the plant $\mathcal{P}$ generates an output word with the same length, i.e., $|w_y|=|\mathcal{R}_p(w_x)|$.  Furthermore $C_3$ also holds on $\mathcal{R}_p$ as  $\mathcal{L}_{e}(\mathcal{P})$ is closed under taking  prefix operation.
\end{proof}

\section{Reactive Supervisory Control Problem}\label{sec:reactive_control} 	 
% The  architecture of reactive supervisory and an open DES is shown in Fig.~\ref{fig:rec_sup_arch}. 
 With the proposed open DES model, we aim to design a supervisor to control the plant with respect to a reactive specification.
 In our setup, the reactive supervisor observes history of  the environment input words, and the plant internal  language, and then chooses a \textit{control pattern}  for the plant. 
 Intuitively control pattern is a set of events that the supervisor permits to be executed  at any given state of the plant. Any control pattern must include the uncontrollable events, since no supervisor can  disable them. 
 Control pattern set is  defined as following and it is illustrated in Example \ref{exp:control_pattern}.{ 
 \begin{definition}
 Set of all control patterns in an open DES $\mathcal{P}$,  is defined as  $\Theta=\{\theta \in 2^{\Sigma_p} \mid \Sigma_u \subseteq \theta \}$. 
 \end{definition}
\begin{example}\label{exp:control_pattern}
Consider the open DES in Fig. \ref{fig:example1} with $\Sigma_{c}=\{\sigma_c\}$, and  $\Sigma_{uc}=\{ \sigma_u\}$. The control patterns are $\theta_1=\{\sigma_c,\sigma_u\}$,  $\theta_2=\{\sigma_u\}$.
 \hfill\ensuremath{\square}
 \end{example}
Another ingredient for the presented reactive supervisory control framework  is the histories of the environment and plant behaviors. We define this history as $His_x(\mathcal{P}):=\mathrm{P}_{xy}(\mathcal{L}_e(\mathcal{P})) \cdot \Sigma_x$ that captures the  environment inputs and the occurred internal events of the plant. 

% \begin{definition}
% Consider a  function  $His:\Sigma^*_e \to (\Sigma_x\cup\Sigma_p)^*$ which is defined inductively as follows: 
% \begin{itemize}
%     \item $His((\epsilon,\epsilon,\epsilon))=\epsilon\epsilon$,
%     \item  $\forall w_e \in \Sigma^*_e$, and $(x,\sigma,y)\in \Sigma_e$\\
%             $His(w_e.(x,\sigma,y))=His(w_e).x\sigma$.  
% \end{itemize}
%Denoted by $His(\mathcal{L}_e(\mathcal{P}))=\{His(w_e)\mid w_e \in  \mathcal{L}_e(\mathcal{P})\}$ is the history of the environment and the plant behaviour.
%\end{definition}
 The reactive supervisor is defined as a function $\mathcal{S}: His_x(\mathcal{P})\to \Theta$.}
 The supervised open DES then only executes a state transition if  it is not restricted by the environment input, and  it is allowed by the supervisor. 
 We  denote the supervised plant  by $\mathcal{S} \slash \mathcal{P}$, and  define its closed-loop languages in the following definition.
 \begin{definition}
The extended input-output language of supervised  plant $\mathcal{L}_e(\mathcal{S} \slash \mathcal{P})$ is recursively defined by:
\begin{itemize}
\item $\left(\epsilon,\epsilon,\epsilon \right) \in \mathcal{L}_e(\mathcal{S} \slash \mathcal{P})$, and
\item for any $w_e \in \mathcal{L}_e(\mathcal{S} \slash \mathcal{P})$, and $(x,\sigma,y) \in  \Sigma_e$ then $w_e(x,\sigma,y) \in \mathcal{L}_e(\mathcal{S} \slash \mathcal{P})$ iff $w_e(x,\sigma,y) \in \mathcal{L}_e(\mathcal{P})$,   and $\sigma \in \mathcal{S}(\mathrm{P}_{xy}(w_e) \cdot x)$.
\end{itemize}
and the marked  set of that is  $\mathcal{L}_{e,m}(\mathcal{S} \slash \mathcal{P})= \mathcal{L}_e(\mathcal{S} \slash \mathcal{P}) \cap   \mathcal{L}_{e,m}( \mathcal{P})$.
The input-output language of  $\mathcal{S} \slash \mathcal{P}$ also  is defined by $\mathcal{L}_{io}(\mathcal{S} \slash \mathcal{P})=\{w \in ( \Sigma_x \times \Sigma_y )^* \mid w \in \mathrm{P}_{xy}(\mathcal{L}_e(\mathcal{S} \slash \mathcal{P}))\}$, and $\mathcal{S} \slash \mathcal{P}$ marked input-output language  is $\mathcal{L}_{io,m}(\mathcal{S} \slash \mathcal{P})=\mathcal{L}_{io}(\mathcal{S} \slash \mathcal{P}) \cap \mathcal{L}_{io,m}( \mathcal{P})$. 
\end{definition}

We aim to design a reactive supervisor with
a pessimistic environment assumption, meaning, the open
DES’s environment is free to behave as it pleases. Therefore,
the input-complete property ensures that the  specification
has defined an output for any environment behaviour.

\begin{definition}
We call  language $K \subseteq (\Sigma_x \times \Sigma_y)^*$ a finite regular reactive specification if it can be realized by an input-complete transducer $\mathcal{T} $ such that  $K = \mathcal{L}_m(\mathcal{T} )$.
\end{definition}

Given a reactive specification $K$ for a plant $\mathcal{P}$, the reactive supervisor task is to control the internal controllable events such that the input-output behavior of  supervised plant meets the specification, $\mathcal{L}_{io,m}(\mathcal{S} \slash \mathcal{P})=K$.

An important extra requirement for supervisory control of DES is a non-blocking property \cite{ramadge1987supervisory}. The  non-blocking condition  ensures that the supervisor does not disable all the controllable events in one state which results  a deadlock state in the supervised plant. This condition conventionally is defined by  $\overline{\mathcal{L}_{m}(\mathcal{S} \slash \mathcal{P})} = {\mathcal{L}(\mathcal{S} \slash \mathcal{P})}$ \cite{ramadge1987supervisory}. 
In open DESs, however, the situation could be more complicated.
\newcommand*{\lc}{\mathcal{S} \slash \mathcal{P}}
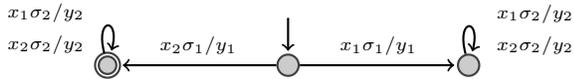
\begin{figure}[!h]  
\centering
\begin{tikzpicture}[shorten >=1pt,node distance=2.4cm,on grid,auto, bend angle=20, thick,scale=1, every node/.style={transform shape}] 
	\node[state_p,initial above,initial text=] (q_0)   {};
     \node[state_p] (q_1)   [left=of q_0,accepting,yshift=0cm] {};
    \node[state_p] (q_2) [right=of q_0,xshift=0cm] {};
	\path[->]
	 (q_0) edge []   node    [above]     {$\scriptstyle x_2\sigma_{1} \slash y_1$} (q_1)
	 (q_0) edge []   node   [above]      {$\scriptstyle x_1\sigma_{1} \slash y_1$} (q_2)
%	 (q_2) edge []   node    [below]     {$\scriptstyle x_2,\sigma_{2} \slash y_2$} (q_3)
	(q_1) edge [loop above]   node     [left,align=center, pos=0.6]     {\begin{tabular}{c}$\scriptstyle x_1\sigma_2\slash y_2$\\$\scriptstyle  x_2\sigma_2\slash y_2$\end{tabular}} (q_1)
	 (q_2) edge [loop above ]   node   [right,align=center, pos=0.6]      {\begin{tabular}{c} $\scriptstyle x_1\sigma_{2} \slash y_2$\\ $\scriptstyle x_2\sigma_{2} \slash y_2$ \end{tabular} } (q_0)
	 
	 ;
			\end{tikzpicture}
	\caption{Open discrete event system for Example \ref{exp:nonblocking}. }
    \label{fig:example2}
		\end{figure} 
{ 	
\begin{example}\label{exp:nonblocking}
Consider the open DES shown in Fig.~\ref{fig:example2}, with  $\Sigma
_c=\Sigma_p$. Let the reactive supervisor allows all the internal events, i.e. $\mathcal{S}(h)=\Sigma_c$ for  all
$h \in His_x(\mathcal{P})$. 
% \begin{align*}
% \mathcal{S}(h)=\left\{
%                 \begin{array}{ll}
%                   \sigma_{1} & h=\epsilon x_1 \vee \epsilon x_1	\sigma_{1}  x_2 \vee \epsilon x_2	\sigma_{1}  x_2	\\
%                   \sigma_{2} & h=\epsilon x_2 \vee \epsilon x_2	\sigma_{2}  x_1 \vee \epsilon x_1	\sigma_{1}  x_1\\
%                  % \sigma_{1} & w= \epsilon x_1	\sigma_{1}  x_1	\\
%                 %  \sigma_{2} & w=\ \epsilon x_1	\sigma_{2}  x_1\\
%                   %\sigma_{2} &  w =\epsilon x_1	\sigma_{1}  x_1	\\
%                   \Sigma_c & h=His(\mathcal{L}_e(\mathcal{P})).
%                 \end{array}
%              \right.
% \end{align*}
The  internal language for supervised plant are $\mathcal{L}(\lc)=\{\epsilon,\sigma_{1},\sigma_{1}\sigma^n_{2} \mid n \in \mathbb{N} \}$, where  $\sigma_{1}\sigma^*_{2}$ is omitted as   $\sigma_{1}\sigma^*_{2}+\sigma_{1}\sigma^*_{2}=\sigma_{1}\sigma^*_{2}$,  and the marked language is $\mathcal{L}_{m}(\lc)=\{\sigma_{1}\sigma^n_{2} \mid n \in  \mathbb{N} \}$. The conventional non-blocking condition holds here, $\overline{\mathcal{L}_{m}(\mathcal{S} \slash \mathcal{P})} = {\mathcal{L}(\mathcal{S} \slash \mathcal{P})}$, however, the environment can force the supervised plant to stay in the livelock by  providing input word $w_x \in \{ x^n_1\mid n \in \mathbb{N}\} $ or $w_x \in \{ x_1x^n_2\mid k \in \mathbb{N}\}$.  
 \hfill\ensuremath{\square}
\end{example} }
As Example~\ref{exp:nonblocking} illustrated, we need a stronger non-blocking definition that can be held for any environment behavior. { 
\begin{definition}\label{def:non-blocking}
A reactive supervisor $\mathcal{S}: His_x(\mathcal{P}) \to \Theta$ is said to be  non-blocking if $   \overline{\mathcal{L}_{e,m}(\lc)}=\mathcal{L}_{e}(\lc)$.
\end{definition}}
The Non-blocking requirement guarantees all  prefix of the supervised plant can be extended to the mark states.  
%It can be shown that the non-blocking property of a reactive supervisor guarantees  all  prefix of the supervised plant can be extended to the mark states. 
{  
\begin{proposition}\label{prop:non-blocking}
A non-blocking reactive supervisor $\mathcal{S}: His_x(\mathcal{P}) \to \Theta$ satisfies  $\overline{\mathcal{L}_{io,m}(\mathcal{S} \slash \mathcal{P})}={\mathcal{L}_{io}(\mathcal{S} \slash \mathcal{P})}$.
\end{proposition}
\begin{proof}
%This proposition is a directed implication of the non-blocking property given in Definition~\ref{def:non-blocking}, since the output function $\gamma_p$ does not effect on non-blocking  property of the reactive supervisor.
It is sufficient to show that if the supervisor is non-blocking, $\overline{\mathcal{L}_{io,m}(\mathcal{S} \slash \mathcal{P})}={\mathcal{L}_{io}(\mathcal{S} \slash \mathcal{P})}$, then ${\mathcal{L}_{io}(\lc)} \subseteq \overline{\mathcal{L}_{io,m}(\lc)}$ holds.
 For any $w \in \mathcal{L}_{io}(\lc)$, there exist $w_e \in \mathcal{L}_{e}(\lc)=\overline{\mathcal{L}_{e,m}(\lc)}$ such that $\mathrm{P}_{xy}(w_e)=w$. % Since the supervisor is non-blocking
 Therefore, since $w_e \in \overline{\mathcal{L}_{e,m}(\lc)}$, there must exist $w'_e \in \mathcal{L}_{e,m}(\lc)$, such that $w_e \in \overline{w'_e}$. Moreover based on the definition of $\mathcal{L}_{io,m}(\lc)$, we have $\overline{\mathrm{P}_{xy}(w'_e)} \subseteq \overline{\mathcal{L}_{io,m}(\lc)}$.
 Hence, $w=\mathrm{P}_{xy}(w_e) \in  \overline{\mathrm{P}_{xy}(w'_e)} \subseteq \overline{\mathcal{L}_{io,m}(\lc)}$.
\end{proof}}

%The control objective a non-blocking reactive supervisor is to  regulate the internal behavior of plant such that the input-output language satisfies a reactive specification. This problem is formally defined as follows. 

Now we are ready to formulate the reactive supervisor control problem for an open DES as follows.  
\begin{problem}
Given a non-empty regular reactive specification ${K} \subseteq \mathcal{L}_{io,m}(\mathcal{P})$, uncontrolled deterministic open DES $\mathcal{P}$, and control pattern set $\Theta$, synthesis a non-blocking reactive supervisor $\mathcal{S}: His_x(\mathcal{P}) \to \Theta$, such that $ \mathcal{L}_{io,m}(\mathcal{S} \slash \mathcal{P})=K$.
\end{problem}

\section{The Existence of Supervisor}\label{sec:reactive_sup_design} 

The first question we ask is under what condition such a supervisor exists. In Ramadge-Wonham supervisory control of DES, the controllability of a specification is a necessary and sufficient condition for the existence of a supervisor \cite{ramadge1989control}. The conventional controllability of a language specification ${K}$, is defined with respect to uncontrollable events $\Sigma_u$, and the plant's internal language $\mathcal{L}(\mathcal{P})$  as $\overline{{K}}(\Sigma_u) \cap \mathcal{L}(\mathcal{P}) \subseteq \overline{{K}}$. Intuitively, this requirement asks if there  is any prefix of the specification, $s \in \overline{K}$, that exists in the uncontrolled plant language  $s \in \mathcal{L}(\mathcal{P})$, and is  followed by an  internal  uncontrollable event $\sigma_u \in \Sigma_u$, $s\sigma_u$ also must be included in the specification, $s\sigma_u \in \overline{K}$, since the supervisor can not disable $\sigma_u$.
In open DES setup however,  the uncontrollability of a reactive specification not only depends on the internal uncontrollable events, but also can be caused by an  input word that enforces  the execution of an undesired output event.
We therefore propose to formally define uncontrollable input-output set, denoted by $\Sigma^{io}_u$, that captures the plant output corresponding to the uncontrollable internal events. 

% Another consideration here is  the uncontrollable outputs may not be unique. For instance, in Example \ref{exp:introduce_openDES}, with $\Sigma_{c}=\{\sigma_1,\sigma_2\}$, and  $\Sigma_{uc}=\{ \sigma_u\}$, output $y_1$ is controllable when the plant is at  $q_2$, and it is uncontrollable at state $q_1$ for input $x_2$. In order to have a unique uncontrollable input-output set, we assume  the following condition on the output function of plant $\mathcal{P}$. 
% \begin{assumption}
% For any $q,q'\in Q_p$, $\sigma_u \in \Sigma_u$, and $\sigma_c \in \Sigma_{c}$, such that $\delta(q,\sigma_c,x)$! and $\delta(q',\sigma_{uc},x)$! for some $x \in \Sigma_x$, $\gamma_p(q,\sigma_{c})\neq \gamma_p(q',\sigma_{uc})$ holds.
% \end{assumption}
% Note that, this assumption is not restrictive, since, we can always assume re-label the output alphabets to obtain unique uncontrollable input-output set.
\begin{definition} 
Consider an open DES $\mathcal{P}=(Q_p,\Sigma_p=\Sigma_c\cup\Sigma_{uc},\Sigma_x,\Sigma_y, q_{p0},\delta_p,\gamma_p,Q_{pm})$, the uncontrollable input-output event set of $\mathcal{P}$ is denoted by $\Sigma^{io}_u$, and is defined by
$\Sigma^{io}_u=\{(x,y) \in \Sigma_x \times \Sigma_y \mid \text{for any } \sigma_u \in \Sigma_u : \delta_p(q,x,\sigma_u)! \text{ and } y = \gamma_p(q,\sigma_u) \text{ for some } q \in Q_p\}$.
\end{definition}

The controllability of a reactive specification is only depend on the plant's output language since the  input component of the specification $\mathrm{P}_{x}(K)$, is generated by the  environment.
%, and  any input word $w_x \in \Sigma^*_x$, by the definition of reactive specification,  should be  included in the  specification, i.e.,   $w_x \in \mathrm{P}_{x}(K)$ for all $ w_x \in \Sigma^*_x$.
Accordingly, we can define the input-output controllable language in the reactive supervisory control setup as following. 
\begin{definition}
A non-empty reactive specification $K \subseteq \mathcal{L}_{io}(\mathcal{P})$ is called output controllable with respect to $\Sigma^{io}_u$, and  $\mathcal{L}_{io}(\mathcal{P})$, if $\overline{{K}}(\Sigma^{io}_u) \cap \mathcal{L}_{io}(\mathcal{P}) \subseteq \overline{{K}}$.
\end{definition}

In addition to the output controllability condition defined above, the  specification is also required to be closed \cite{ramadge1989control}.{ 
\begin{definition}
A reactive specification ${K}$ is called closed with respect to  $\mathcal{L}_{io,m}(\mathcal{P})$, in short $\mathcal{L}_{io,m}(\mathcal{P})$-closed, if   $K=\overline{{K}} \cap  \mathcal{L}_{io,m}(\mathcal{P})$. 
\end{definition}}

The following lemma shows output controllability and $\mathcal{L}_{io,m}(\mathcal{P})$-closedness are  necessary conditions for the existence of a non-blocking reactive supervisory control.

 \begin{lemma} %{[Necessary Condition]} 
 \label{lem:nec_cond}
 If there exist a non-blocking reactive supervisor $\mathcal{S}: His_x(\mathcal{P}) \to \Theta$ such that $\mathcal{L}_{io,m}(\mathcal{S} \slash \mathcal{P})={K}$, then ${K}$ is output controllable and  $\mathcal{L}_{io,m}(\mathcal{P})$-closed.
 \end{lemma}
{ 
\begin{proof}
Since the supervisor is non-blocking, by  Proposition~\ref{prop:non-blocking},  we have $\overline{\mathcal{L}_{io,m}(\mathcal{S} \slash \mathcal{P})}={\mathcal{L}_{io}(\mathcal{S} \slash \mathcal{P})}$, and therefore 
\begin{align*}
    \overline{{K}} \cap  \mathcal{L}_{io,m}( \mathcal{P})&=\overline{\mathcal{L}_{io,m}(\mathcal{S} \slash \mathcal{P})} \cap  \mathcal{L}_{io,m}( \mathcal{P})\\
    &= \mathcal{L}_{io}(\lc) \cap  \mathcal{L}_{io,m}( \mathcal{P})
\end{align*}
which implies $\overline{{K}} \cap  \mathcal{L}_{io,m}( \mathcal{P})=\mathcal{L}_{io,m}( \mathcal{S} \slash \mathcal{P})=K$. Hence, $K$ is $\mathcal{L}_{io,m}(\mathcal{P})$-closed.
To prove the controllability condition of ${K}$, let's consider any $w \in \overline{{K}}$, and $(x,y) \in \Sigma^{io}_u$ such that $w(x,y) \in \overline{{K}}\Sigma^{io}_u \cap  \mathcal{L}_{io}(\mathcal{P})$. By the existence of the non-blocking supervisor, we have  $\overline{K}=\overline{\mathcal{L}_{io,m}(\lc)} =\mathcal{L}_{io}(\lc)$.
Therefore $w \in \mathcal{L}_{io}(\lc)$, and then there should exist $w_e \in \mathcal{L}_{e}(\mathcal{S} \slash \mathcal{P})$ such that $\mathrm{P}_{xy}(w_e)=w$.  Given the history of the extended word $w_e$, and the environment input $x$, 
we have an enabled internal event by the supervisor, i.e., $ \sigma \in \mathcal{S}(\mathrm{P}_{xp}(w_e) \cdot  x)$ such that $y=\gamma_p(q_{p0},\mathrm{P}_p(w_e)\cdot \sigma)$. 
Note that $\mathrm{P}_{xy}(w_e)\cdot x$ captures  history of the executed internal  and environment events in $w_e$ concatenated with $x$.
Thus, $w_e \cdot (x,\sigma,y) \in \mathcal{L}_{e}(\lc)$ which implies that $w \cdot (x,y) \in \mathcal{L}_{io}(\mathcal{S} \slash \mathcal{P})= \overline{K}$. As a result, $K$ is output controllable. 
\end{proof}}

\subsection{Game-based Reactive Supervisor Design }

To gain insights on the sufficient conditions of the existence of the supervisory, we approach through a constructive way. Particularly, inspired by the reactive synthesis work, we present a game-theoretic  method to design such a non-blocking reactive supervisor to achieve the reactive specification. 

Similar to the conventional  reactive synthesis, the designing of a  controller can be seen as a game between the reactive supervisor $\mathcal{S}$, as a player aiming to control the plant  to satisfy the specification $K$, and the environment with intention of driving the plant to violate $K$. 
However, in reactive supervisory control of open DESs, we have a plant that is capable of restricting the supervisor with internal uncontrollable events. 
More specifically, in this setup when the environment selects an input, the supervisor chooses a control pattern that  must include all the uncontrollable events, and then the plant can only execute  an internal event from the assigned control pattern set.  
We therefore propose to design a two-player turned-based  game between  the  environment and the reactive  supervisor.
The arena construction has the following steps:

(a) Given $\mathcal{P}=(Q_p,\Sigma_p,\Sigma_x,\Sigma_y, q_{p0},\delta_p,\gamma_p,Q_{pm})$, and  a non-empty  reactive specification $K \subseteq \mathcal{L}_{io,m}(\mathcal{P})$. Construct a  transducer $\mathcal{T}_K=(Q_k,  \Sigma_x, \Sigma_y , q_{k0},\delta_k, \gamma_k,Q_{km})$, such that $\mathcal{L}_{m}(\mathcal{T}_K)=K$. 

(b) Construct an  automaton   $A_K=(\tilde{Q}_k , \Sigma_x\times  \Sigma_y,  \tilde{q}_{k0}, \tilde{\delta}_k, \tilde{Q}_{km})$, where  $\tilde{Q}_k = {Q_k}\cup \{q_{\bot}\} $, and  $q_{\bot}$ is a dummy state, $\tilde{q}_{k0}=q_{k0}$, $\tilde{Q}_{km}=Q_{km} \subseteq \tilde{Q}_k$, and the transition function for any $q \in \tilde{Q}_k$,  $x\in \Sigma_x$, and $y \in \Sigma_y$ is 
$\tilde{\delta}_{k}(q,(x,y)) =  \delta_{k}(q,x)  $ if $q \in Q_k \wedge \delta_{k}(q,x)!  \wedge y=\gamma_{k}(q,x)$, and otherwise $\tilde{\delta}_{k}(q,(x,y)) =q_{\bot}$.
% \begin{align*}
%   \tilde{\delta}_{k}(q,(x,y)) = \left\{  \begin{aligned}&\delta_{k}(q,x)  &\text{if } q \in Q_k \wedge \delta_{k}(q,x)! \\&   & \wedge y=\gamma_{k}(q,x)\\ &q_{\bot} &\text{Otherwise} \end{aligned} \right.  
% \end{align*} 
Therefore, for any $w \in (\Sigma_x \times \Sigma_y)^*$, we have $w\not \in  \overline{K}$, if and only if $\tilde{\delta}_{k}(\tilde{q}_{k0},w)=q_\bot$.

(c) Construct the  synchronous composition $\mathcal{P} \parallel A_K =(Q_s,\Sigma_p,\Sigma_x,\Sigma_y, q_{s0},\delta_s,\gamma_s,Q_{sm})$, where  $Q_s = Q_p \times \tilde{Q}_k  $,  $q_{s0}=(q_{p0},\tilde{q}_{k0})$. The transition function  $\delta_{s}:Q_s \times \Sigma_x \times \Sigma_p \to Q_s$, for any $q_p \in Q_p$, $\tilde{q}_k \in \tilde{Q}_k$, $(q_p,\tilde{q}_k) \in Q_s$ , $\sigma \in \Sigma_p$, $x\in\Sigma_x$, and $y\in\Sigma_y$ is  defined as: $\delta_{s}((q_p,\tilde{q}_k),x,\sigma) =( \delta_{p}(q_p,x,\sigma), \tilde{\delta}_{k}(\tilde{q}_k,(x,y)))$   if  $\delta_{p}(q_p,x,\sigma)! \wedge \gamma(q_p,\sigma)=y \wedge \tilde{\delta}_{k}(\tilde{q}_k,(x,y))!$, and otherwise it is undefined. The output function  is $\gamma_s((q_p,\tilde{q}_k),\sigma)=\gamma_p(q_p,\sigma)$, and the marking states are $Q_{sm}= Q_{pm} \times \tilde{Q}_{km}$.
% \begin{align*}
%  &\delta_{s}((q_p,\tilde{q}_k),\sigma,x) =\\
%  &\left\{  \begin{aligned}& ( \delta_{p}(q_p,\sigma,x), \tilde{\delta}_{k}(\tilde{q}_k,(x,y))) &  \text{if } \delta_{p}(q_p,\sigma,x)!\\& & \wedge \gamma(q,\sigma)=y\\&\text{Otherwise} & \text{Undefined} \end{aligned}
% \right.
% \end{align*} 
%We therefore have  $\mathrm{P}_{xp}(\mathcal{L}_{e}(\mathcal{P} \parallel A_K ))=\mathrm{P}_{xp}(\mathcal{L}_{e}(\mathcal{P}))$, and $\mathcal{L}_{io,m}(\mathcal{P} \parallel A_K)= K$.

%We therefore have  $\mathcal{L}_{io,m}(\mathcal{P} \parallel A_K)= K$.

Formally, given the control pattern set $\Theta$, and  $\mathcal{P} \parallel A_K$,  the  arena is defined as $\mathcal{G}=\left(Q_g, \Sigma_g,v_{g0},\delta_g,win\right)$, where  $Q_{g}$ is state of the game, which  is partitioned into two disjoint set of $V_e$, and $V_s$ respectively representing  the environment, supervised  plant states. Let's first denote the state set $Q_s$ by $\{q_1,q_2,\dots,q_{|Q_s|} \}$, and then for all $q_i \in Q_s$, $x \in \Sigma_x$, and $\theta \in \Theta$, the game states are defined by:
\begin{itemize}
    \item $V_e:=\{q_i \mid  i\in \{1,\dots,{|Q_s|} \}  $,
    \item $V_s:=\{(q_i,x) \mid i\in \{1,\dots,{|Q_s|} \}\} $.
   % \item $V_p:=\{(q_i,x,\theta) \mid i\in \{1,\dots,{|Q_s|} \}\}$.
\end{itemize}
%Here $v_\bot$ is a dummy state. 
The initial state of the game is $v_{g0}=q_{s0}$, and $q_{g0}\in V_e$. The  event set is $\Sigma_{g} := \Sigma_x \cup (  \Theta \times \Sigma_p)$.
The transition function is defined as $\delta_g:=  f_e \cup f_s $ where $f_e \subseteq V_e \times \Sigma_x \times V_s$,  $f_s \subseteq V_s \times (  \Theta \times \Sigma_p) \times V_e$. 
They are defined as:

\begin{itemize}
    \item $f_e(q,x)=\{(q,x)  \mid q \in Q_s, x\in \Sigma_x, (q,x) \in V_s\}$,
    \item $f_s((q,x),(\theta,\sigma))=\{q' \mid (q,x) \in V_s, \theta \in \Theta,  \sigma \in \theta, \delta_s(q,x,\sigma)=q'\}$.
 %   \item $f_p((q,x,\theta),\sigma)=\{q' \mid q' \in V_e, (q,x,\theta) \in V_p, \sigma \in \theta, \delta_s(q,\sigma,x)=q'\}.$
\end{itemize}
Here, the game state $q \in V_e$ is the environment player state, and she can selects an input $x \in \Sigma_x$.  Given the environment input $x$, the game proceeds to the supervisor player state $(q,x)$. The supervisor then chooses a control pattern $\theta$ from $\Theta$, and the supervised plant selects a permitted internal action, i.e., $\sigma \in \theta$.  At this state the game again proceeds to the  environment state $q'$ if  $q'=\delta_s(q,x,\sigma)!$.
In this setup, the transition function $f_e$ captures the transitions according to all possible subset of environment input, and  $f_s$ represents the transitions according to supervised plant choice.

Let's define a set of state  $V_\bot=\{(q_p,q_\bot) \mid (q_p,q_\bot) \in V_e, q_p \in Q_p, q_\bot \in \tilde{Q}_k \} $, representing the loosing states for the reactive supervisor. 
The winning condition is then defined as a safety condition $win= Q_g - V_{\bot}$. 
The following example illustrates the game construction.

\begin{example}\label{exp:game_const}
Consider the plant model in Fig.~\ref{fig:example1}. Let's assume $\Sigma_{uc}=\{\sigma_{u} \}$, and $\Sigma_{c}=\{\sigma_{c}\}$. 
The reactive specification is considered as 
%$K=\{(x_1,y_1)(x_2,y_1),  (x_1,y_1)(x_1,y_2), (x_2,y_2)(x_2,y_1), (x_2,y_2)(x_1,y_2)\}$.
\begin{align*}
    K=\{&(x_1,y_1)(x_2,y_1),  (x_1,y_1)(x_1,y_2),\\
    &(x_2,y_2)(x_2,y_1), (x_2,y_2)(x_1,y_2)\}.
\end{align*}

The   arena for the plant  and the reactive specification $K$ can be constructed as shown in Fig.~\ref{fig:game_arena_exmp}. The circles  are the environment states, the squares are the supervised plant states. The loosing states is $v_\bot$, and all other states are winning states. 
%Please note that for clarity of the presentation, we have removed  any transitions from $q_1,q_0$ that leads to loosing state.%, and also the runs over $\{x_1x_2\}^*$.
%If the labels is not shown, it follows the plant $\mathcal{P}$ output function $\gamma_p(q,\sigma)$ 
 \hfill\ensuremath{\square}
\end{example}

\tikzstyle{game_env}=[circle,thick,draw=black!75,
  			  fill=black!20,minimum size=3mm,inner sep=.5mm]
\tikzstyle{game_sup}=[rectangle,thick,draw=black!75,
  			  fill=black!20,minimum size=3mm,inner sep=.5mm]
 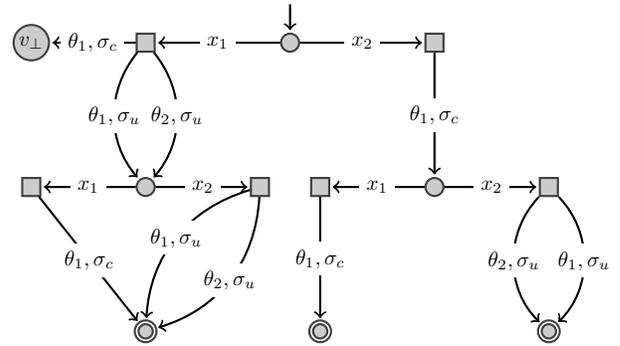
\begin{figure}[!ht]  
 	\centering
 	\begin{tikzpicture}[shorten >=1pt,node distance=.9cm,on grid,auto, bend angle=20, thick,scale=.8, every node/.style={transform shape}] 
%%%%%%%%%%%%%%%%%%%%% Environment Layer 1 Nodes %%%%%%%%%%%%%%%%%%%%%% 	
 	\node[game_env,initial above,initial text=] (ve_q0s0)   {};
%%%%%%%%%%%%%%%%%%%%% Supervisor Layer 1 Nodes %%%%%%%%%%%%%%%%%%%%%% 	 	
    \node[game_sup] (vs_q0s0x1)   [left=of ve_q0s0,yshift=0cm,xshift=-1.5cm] {};
    \node[game_sup] (vs_q0s0x2)   [right=of ve_q0s0,yshift=0cm,xshift=1.5cm] {};
%%%%%%%%%%%%%%%%%%%%% Environment Layer 2  Nodes %%%%%%%%%%%%%%%%%%%%%% 
   	\node[game_env] (ve_q1qN) [left=of vs_q0s0x1,yshift=0cm,xshift=-1.0cm]  {$v_\bot$};%{$q_1,q_\bot$}
   	\node[game_env] (ve_q1s1) [below=of vs_q0s0x1,yshift=-1.5cm,xshift=0cm]  {};%{$q_1,s_1$}
   	\node[game_env] (ve_q2s2) [below=of vs_q0s0x2,yshift=-1.5cm,xshift=0cm]  {};%{$q_2,s_2$}
%%%%%%%%%%%%%%%%%%%%% Supervisor Layer 2 Nodes %%%%%%%%%%%%%%%%%%%%%% 
    \node[game_sup] (vs_q1s1x1)  [left=of ve_q1s1,yshift=0cm,xshift=-1.0cm] {};
    \node[game_sup] (vs_q1s1x2)  [right=of ve_q1s1,yshift=0cm,xshift=1.0cm] {};
    
    \node[game_sup] (vs_q2s2x1)  [left=of ve_q2s2,yshift=0cm,xshift=-1.0cm] {};
    \node[game_sup] (vs_q2s2x2)  [right=of ve_q2s2,yshift=0cm,xshift=1.0cm] {};
%%%%%%%%%%%%%%%%%%%%% Environment Layer 3 Nodes %%%%%%%%%%%%%%%%%%%%%%  
    \node[game_env] (ve_q2s3) [accepting,below=of ve_q1s1,yshift=-1.5cm,xshift=0cm]  {};
    \node[game_env] (ve_q2s6) [accepting,below=of vs_q2s2x1,yshift=-1.5cm,xshift=0cm]  {};
     \node[game_env] (ve_q2s5) [accepting,below=of vs_q2s2x2,yshift=-1.5cm,xshift=0cm]  {};
   
 \path[->]
%%%%%%%%%%%%%%%%%%%%% Environment Layer 1 Edges %%%%%%%%%%%%%%%%%%%%%% 
 (ve_q0s0) edge [left=0] node [midway,fill=white,anchor = center]  {$x_1$} (vs_q0s0x1)
 (ve_q0s0) edge [left=0] node [midway,fill=white,anchor = center]  {$x_2$} (vs_q0s0x2)
 %%%%%%%%%%%%%%%%%%%%% Supervisor Layer 1 Edges %%%%%%%%%%%%%%%%%%%%%% 	
 (vs_q0s0x1) edge [right=-45] node [midway,fill=white,anchor = center]  {$\theta_1,\sigma_c$} (ve_q1qN)
 (vs_q0s0x1) edge [right=0, bend right=40] node [midway,fill=white,anchor = center]  {$\theta_1,\sigma_u$} (ve_q1s1)
 (vs_q0s0x1) edge [right=0, bend left=40] node [midway,fill=white,anchor = center]  {$\theta_2,\sigma_u$} (ve_q1s1)
 (vs_q0s0x2) edge [right=0] node [midway,fill=white,anchor = center]  {$\theta_1,\sigma_c$} (ve_q2s2) 
 %%%%%%%%%%%%%%%%%%%%% Environment Layer 2 Edges %%%%%%%%%%%%%%%%%%%%%%
 (ve_q1s1) edge [left=0] node [midway,fill=white,anchor = center]  {$x_1$} (vs_q1s1x1)
 (ve_q1s1) edge [left=0] node [midway,fill=white,anchor = center]  {$x_2$} (vs_q1s1x2)
 
 (ve_q2s2) edge [left=0] node [midway,fill=white,anchor = center]  {$x_1$} (vs_q2s2x1)
 (ve_q2s2) edge [left=0] node [midway,fill=white,anchor = center]  {$x_2$} (vs_q2s2x2)
 %%%%%%%%%%%%%%%%%%%%% Supervisor Layer 2 Edges %%%%%%%%%%%%%%%%%%%%%%
  (vs_q1s1x1) edge [right=0] node [midway,fill=white,anchor = center]  {$\theta_1,\sigma_c$} (ve_q2s3) 
  (vs_q1s1x2) edge [bend right=35] node [midway,fill=white,anchor = center]  {$\theta_1,\sigma_u$} (ve_q2s3) 
  (vs_q1s1x2) edge [bend left=35] node [midway,fill=white,anchor = center]  {$\theta_2,\sigma_u$} (ve_q2s3) 
  
  (vs_q2s2x1) edge [] node [midway,fill=white,anchor = center]  {$\theta_1,\sigma_c$} (ve_q2s6)
  (vs_q2s2x2) edge [bend left=45] node [midway,fill=white,anchor = center]  {$\theta_1,\sigma_u$} (ve_q2s5)
  (vs_q2s2x2) edge [bend right=45] node [midway,fill=white,anchor = center]  {$\theta_2,\sigma_u$} (ve_q2s5)
 ;
 	\end{tikzpicture}
 \caption{Game arena constructed from the open DES in Fig.~\ref{fig:example1}, the reactive specification in Example~\ref{exp:game_const}, and  the control pattern in Example~\ref{exp:control_pattern}. }
\label{fig:game_arena_exmp}
 \end{figure}
 
%The event $\sigma_\bot$ and transition function $\delta_{\bot}$ captures  the DES enters  deadlocks.

% The reactive supervisor wins the game from the initial state $v_{g0}$ if it can force the game states  to always stay in set $win$, regardless how the environment and the plant choose their output transition \cite{ummels2006rational}.

\subsection{Solving Safety Games}
Each player in $\mathcal{G}$ starts from the initial state, $v_{g0}$, and plays the game by observing the history of the states, and when the current state of the game is his state, he can select an outgoing transition. We call a player wins the game $\mathcal{G}$, if he can force the state of the game to stay in set $win$, regardless how the other player plays the game. Two-player safety game $\mathcal{G}$  is determined with  positional  winning strategies and it 
can be solved with simple fixed point construction, with complexity linear time in the size of $|\delta_g|+|Q_g|$. If a player has a winning strategy, a finite realization of it can be obtained \cite{berwanger2013graph}. 
%In the constructed game a deadlock occurs when the supervised plant enters a state that there is  no enabled event  or when the supervisor selects a control pattern that disables all possible events at the state. Therefore, the reactive supervisor

If the reactive supervisor has wining strategy from the initial state $v_{g0}$, it induces a game arena $\mathcal{G}$, that all the runs starting from $v_{g0}$ stays in set $win$ \cite{berwanger2013graph}.
The set of runs over ${\mathcal{G}}$ produce all the extended input-output behavior of the  plant that are compatible with the reactive specification  $K$. Formally, let's denote it as $\mathcal{C}_k({\mathcal{G}}) \subseteq \mathcal{L}_{e}( \mathcal{P})$, that inductively is defined as:

\begin{itemize}
\item $(\epsilon,\epsilon,\epsilon ) \in \mathcal{C}_k({\mathcal{G}})$, $q=\delta_p(q_{p0},\epsilon,\epsilon)$, and 
\item for any $ w_e \in \mathcal{C}_k({\mathcal{G}})$,  $ x, \in \Sigma_x $, $ y \in \Sigma_y$, $ \sigma \in \Sigma_{p}$,\\
then $w_e \cdot (x,\sigma,y) \in \mathcal{C}_k({\mathcal{G}})$ and $q_e=\delta_p(q_{p0},\mathrm{P}_{xp}(w_e))$  iff   $f_{e}(q_e,x)!$, and 
   there exists $ \theta \in \Theta $, and $\sigma \in \Sigma_p$ such that $f_{s}((q_e,x),(\theta,\sigma))!$, $\sigma \in \theta$, and $y=\gamma_p(q_e,\sigma)$.
\end{itemize}

The notation of realizability of a specification in reactive synthesis formalism is the existence of an reactive system that satisfies the specification for any environment behavior. Here we adapt this definition to the existence of a reactive supervisor $\mathcal{S}$ that characterizes the set $\mathcal{C}_k({\mathcal{G}})$ and can extend the history of the game to the mark states. 
The $\mathcal{S} \mathcal{P}$-realizable notation is given as follows.%in Definition~\ref{def:sp-realizable}.
{ 
\begin{definition} \label{def:sp-realizable}
A reactive specification $K \subseteq \mathcal{L}_{io,m} $  is called realizable by a reactive supervisor  $\mathcal{S}: His_x(\mathcal{P})\to \Theta$,  with respect to $\mathcal{L}_{e}(\mathcal{P})$, in short $\mathcal{S} \mathcal{P}$-realizable, if $\mathcal{C}_k({\mathcal{G}}) \neq \emptyset$, and for any $w_e \in \mathcal{C}_k({\mathcal{G}})$, we have 
$\mathrm{P}_{xy}(w_e) \in \overline{K} $ that can be extended to the mark states, i.e., there exists a $w'_e \in \mathcal{L}_{e}(\mathcal{P})$, such that $w_e \in \overline{w'_e}$ and $\mathrm{P}_{xy}(w'_e ) \in K $.
\end{definition}
}
%If the reactive supervisor player wins the game, it 
%Let's define a induced game arena,  $\mathcal{G}_s=(\tilde{Q}_g,\tilde{\delta_g},\Sigma_g,v_{g0},win)$.
% \begin{lemma}
% If the reactive supervisor can force the safety winning  condition in $\mathcal{G}$,  $\mathcal{C}_k({\mathcal{G}})$ satisfies the $\mathcal{S} \mathcal{P}$-realizability conditions, and  rective specification $K \subseteq \mathcal{L}_{io,m} $ is $\mathcal{S} \mathcal{P}$-realizability.
% \end{lemma}
% \begin{proof}
 
% \end{proof}

%The $\mathcal{S} \mathcal{P}$-realizability condition ensures that the environment can not force the  supervised plant to block the 
We then show in Theorem~\ref{thm:nec_suff_cond} that $\mathcal{S} \mathcal{P}$-realizability, together with the output controllability and language closeness, presents a necessary and sufficient condition for the existence of a non-blocking reactive supervisor.
%The proof of Theorem~\ref{thm:nec_suff_cond} is inspired by   \cite[Theorem~1]{ushio2016nonblocking} on the non-blocking supervisor of mealy automata. 
\begin{theorem} \label{thm:nec_suff_cond}
Let a non-empty reactive specification be $K \subseteq \mathcal{L}_{io,m}(\mathcal{P})$. There exits a  reactive supervisor $\mathcal{S}: His_x(\mathcal{P}) \to \Theta$ satisfies $\mathcal{L}_{io,m}(\mathcal{S} \slash \mathcal{P}) = K$ without blocking if and only if $K$ is output controllable, $\mathcal{L}_{io,m}(\mathcal{P})$-closed, and $\mathcal{C}_k({\mathcal{G}})$ is $\mathcal{S} \mathcal{P}$-realizable. 
\end{theorem}
\begin{proof}
 Necessity: since it is assumed there exist a reactive supervisor that satisfies the specification, $\mathcal{L}_{io,m}(\mathcal{S} \slash \mathcal{P}) = K$, according to Lemma~\ref{lem:nec_cond}, $K$ is output controllable and language close. Furthermore it is $\mathcal{S} \mathcal{P}$-realizable by the definition.
 \\
 Sufficiency: 
 Here we need to prove if $K$ is output controllable,  closed with respect to $\mathcal{L}_{io,m}(\mathcal{P})$, and $\mathcal{S} \mathcal{P}$-realizable then there exist a non-blocking reactive supervisor $\mathcal{S}$ such that $\mathcal{L}_{io,m}(\mathcal{S} \slash \mathcal{P}) = K$.
We first prove $\mathcal{L}_{io}(\lc) = \overline{K}$. We know $\mathcal{L}_{io}(\lc) \cap \overline{K} \neq \emptyset$ since $K$ is non-empty,  and by definition of
$(\epsilon,\epsilon,\epsilon) \in \mathcal{L}_{e}(\lc) $. Therefore, we have $(\epsilon,\epsilon) \in \mathcal{L}_{io}(\lc) \cap \overline{K}$. Now, consider $w \in \mathcal{L}_{io}(\lc) \cap \overline{K}$, and $(x,y)\in \Sigma_x\times \Sigma_y$.
Let's first assume $w \cdot (x,y) \in  \mathcal{L}_{io}(\lc)$,  and  $(x,y) \in \Sigma^{io}_u$, then since $K$ is output controllable, $\overline{K}(\Sigma^{io}_u) \cap \mathcal{L}_{io}(\lc) \subseteq \overline{K}$, and therefore, we have $w \cdot (x,y) \in \overline{K}$. 
Now let's consider if $(x,y) \in \{(\Sigma_x \times \Sigma_y) - \Sigma^{io}_u$\},  then there exists $w_e \sigma_e \in \mathcal{L}_{e}(\mathcal{S} \slash \mathcal{P})$ such that $\mathrm{P}_{xy}(w_e)=w$, $\mathrm{P}_{p}(\sigma_e)=\sigma_c$,  which is is permitted by $\mathcal{S}$, i.e., $\sigma_c \in \mathcal{S}(\mathrm{P}_{xp}(w_e) \cdot \mathrm{P}_{x}(\sigma_e))$, and generates output $y=\gamma_p(\delta_p(q_{p0},\mathrm{P}_{xp}(w_e),\sigma_c)$. Since $\sigma_c$ is permitted by $\mathcal{S}$, it should be included in the induced game arena $\mathcal{G}$, and therefore by  $\mathcal{S} \mathcal{P}$-realizability of $K$, we have $w_e \cdot \sigma_e \in  \mathcal{C}_k({\mathcal{G}})$ and $\mathrm{P}_{xy}(w_e \cdot \sigma_e)\in \overline{K}$. Thus  $\mathcal{L}_{io}(\mathcal{S} \slash \mathcal{P}) \subseteq \overline{K}$. 
Now let's consider  $w \cdot (x,y) \in \overline{K}$.
By definition of $A_K$, we have $w \cdot (x,y) \in \overline{\mathcal{L}_m(A_K)}$, and it should belong to the runs over the induced game arena $\mathcal{C}_k({\mathcal{G}})$. 
Since  $\mathcal{C}_k({\mathcal{G}})$ satisfies the $\mathcal{S} \mathcal{P}$-realizability conditions, there exists $w_e  \in \mathcal{C}_k({\mathcal{G}}) \subseteq \mathcal{L}_{e}(\lc)$ such that $\mathrm{P}_{xy}(w_e )=w \cdot (x,y)$, that implies   
$ w \cdot (x,y) \in \mathcal{L}_{io}(\mathcal{S} \slash \mathcal{P})$. Therefore $\mathcal{L}_{io}(\mathcal{S} \slash \mathcal{P}) =  \overline{K}$.
Moreover, since  $K$ is closed with respect to $\mathcal{L}_{io,m}(\mathcal{P})$, we have $\mathcal{L}_{io,m}(\mathcal{S} \slash \mathcal{P})=\mathcal{L}_{io}(\mathcal{S} \slash \mathcal{P}) \cap \mathcal{L}_{io,m}( \mathcal{P})=K$.
 To prove non-blocking property of the reactive supervisor, ${\mathcal{L}_{e}(\lc)} \subseteq \overline{\mathcal{L}_{e,m}(\lc)}$, let's consider any $w_e \in \mathcal{L}_{e}(\mathcal{S} \slash \mathcal{P})$. By $\mathcal{S} \mathcal{P}$-realizability conditions, we have $w_e \in \mathcal{C}_k({\mathcal{G}}) \subseteq \mathcal{L}_{e}(\lc) $, such that $\mathrm{P}_{xy}(w_e) \in \overline{K}=\mathcal{L}_{io}(\mathcal{S} \slash \mathcal{P})  $.  Then there exists a $w'_e \in \mathcal{L}_{e}(\mathcal{S} \slash \mathcal{P})$ such that $w_e \in \overline{w'_e}$, and $\mathrm{P}_{xy}(w'_e) \in {K}$, and therefore, we have $w'_e \in \mathcal{L}_{e,m}(\lc)$. Hence,  ${\mathcal{L}_{e}(\lc)} \subseteq \overline{\mathcal{L}_{e,m}(\lc)}$.
\end{proof}

{ 
\begin{example} \label{exp:sp_realizability}
Consider the game in Fig.~\ref{fig:game_arena_exmp}, and the reactive specification and control patterns  given in Example~\ref{exp:game_const}. The reactive supervisor can win the game and force the plant to stay at set $win$ by choosing the following control pattern. 
%a prefix example in set $\mathcal{C}_k({\mathcal{G}})$ is $\{(\epsilon,\epsilon,\epsilon), (x_1,\sigma_1,y_1) (x_1,\sigma_u,y_2) \} \in \mathcal{C}_k({\mathcal{G}})$.
\begin{align*}
\mathcal{S}(h)=\left\{
                \begin{array}{ll}
                  \theta_2 &  h= x_1 \vee x_1 \sigma_u x_2 \vee x_2 \sigma_c x_2\\
                   \theta_1 & h=  x_2 \vee x_1 \sigma_u x_1 \vee x_2 \sigma_c x_1
                \end{array}
             \right.
\end{align*}
This reactive supervisor can win the game, and hence the graph induced by this supervisor satisfies the $\mathcal{S} \mathcal{P}$-realizability conditions. As it is clear from Fig.~\ref{fig:game_arena_exmp},  all the runs starting from initial state $v_{g0}$, can reach the mark state $q_2$ and never visit $V_\bot$ states.
 \hfill\ensuremath{\square}
\end{example}}

\section{Conclusion}
In this paper, we studied a reactive synthesis problem for open DESs with finite regular language specifications. The proposed DES is a reactive system with a deterministic internal behavior  and an input function.
The control objective for the supervisor is to  achieve the specification regardless of how the  environment behaves.
We provided a necessary and sufficient conditions for  existence of the reactive supervisor.
%If the condition holds, a three-player game arena can be constructed to drive a finite state representative the supervisor. 
%ali talk about the connection sup control and rec synthesis
In our future work, we plan to study  reactive supervisory control with other specifications such as $\omega-$regular  language and temporal logic formulas. 
%Moreover, we will extend this study to open DESs with non-deterministic transition function and satisfaction with respect to bi-simulation.

\bibliographystyle{IEEEtran}        % Include this if you use bibtex
\bibliography{Reference_RSC}
\end{document}